\documentclass[conference]{IEEEtran}
\usepackage{amsmath}    % need for subequations
\usepackage{graphics,epsfig,subfigure}    % need for figures
\usepackage{verbatim}   % useful for program listings
\usepackage{color}      % use if color is used in text
\usepackage{subfigure}  % use for side-by-side figures
\usepackage{amssymb}
\usepackage{epstopdf}
\usepackage{graphicx,cite}
\usepackage{mathrsfs}
\usepackage{url}

\makeatletter
\def\ps@headings{%
\def\@oddhead{\mbox{}\scriptsize\rightmark \hfil \thepage}%
\def\@evenhead{\scriptsize\thepage \hfil \leftmark\mbox{}}%
\def\@oddfoot{}%
\def\@evenfoot{}}
\makeatother
\newtheorem{theorem}{\textbf{Theorem}}[section]
\newtheorem{lemma}[theorem]{\textbf{Lemma}}

\newtheorem{definition}[theorem]{\textbf{Definition}}

\begin{document}

\IEEEoverridecommandlockouts

\title{When Queueing Meets Coding: Optimal-Latency Data Retrieving Scheme in Storage Clouds}

\author{Shengbo Chen\\
 Department of ECE\\
 Ohio State University\\
 chens@ece.osu.edu\\
\and
Yin Sun\\
 Department of ECE\\
The Ohio State University\\
sunyin02@gmail.com\\
\and
Ula\c{s} C. Kozat\\
DOCOMO Innovations, Inc.\\
Palo Alto, CA, 94304\\
kozat@docomoinnovations.com\\
\and
Longbo Huang\\
IIIS, Tsinghua University\\
longbohuang@tsinghua.edu.cn\\
\and
Prasun Sinha\\
 Department of CSE\\
 Ohio State University\\
prasun@cse.ohio-state.edu\\
\and
Guanfeng Liang\\
DOCOMO Innovations, Inc.\\
 Palo Alto, CA, 94304\\
gliang@docomoinnovations.com\\
\and
Xin Liu\\
 Department of CS\\
UC Davis\\
liu@cs.ucdavis.edu\\
\and
Ness B. Shroff\\
 Department of ECE and CSE\\
The Ohio State University\\
shroff@ece.osu.edu\\
}

% \author{\IEEEauthorblockN{ Shengbo Chen\IEEEauthorrefmark{1},
% }
%
%
%
% \IEEEauthorblockA{\IEEEauthorrefmark{1}Department of ECE, The Ohio State University}
% % \IEEEauthorblockA{\IEEEauthorrefmark{2}\IEEEauthorrefmark{3}Department of CSE, The Ohio State University}
% % {Email: \{\IEEEauthorrefmark{1}chens,\IEEEauthorrefmark{2}shroff\}@ece.osu.edu,
% % \IEEEauthorrefmark{3}prasun@cse.ohio-state.edu
% % }
% }

%\IEEEpubid{978-1-4799-3360-0/14/\$31.00 ©2014 IEEE}

\maketitle

%&latex

\begin{abstract}

Storage clouds, such as Amazon S3,  are being widely used for web services and Internet applications. It has been observed that the delay for retrieving data from and placing data into the clouds is quite random, and exhibits weak correlations between different read/write requests. This inspires us to investigate a key problem: can we reduce the delay by  transmitting data replications in parallel or using powerful erasure codes?

In this paper, we study the problem of reducing the delay of  downloading data from cloud storage systems by leveraging multiple parallel threads, assuming that the data  has been encoded and stored in the clouds using fixed rate forward error correction (FEC) codes with parameters $(n,k)$. That is, each file is divided into $k$ equal-sized chunks, which are then expanded into $n$ chunks such that any $k$ chunks out of the $n$ are sufficient to successfully restore the  original file. The model can be depicted as a multiple-server queue with arrivals of data retrieving requests and a server corresponding to a thread. However, this is not a typical queueing model because a server can terminate its operation, depending on when other servers complete their service (due to the redundancy that is spread across the threads). Hence, to the best of our knowledge, the  analysis of this queueing model remains quite uncharted.

Recent traces from Amazon S3 show that the time to retrieve a fixed size chunk is random and can be {approximated as a constant delay plus an \textit{i.i.d.} exponentially distributed random variable. For the tractability of the theoretical analysis, we assume that the chunk downloading time is \textit{i.i.d.} exponentially distributed. Under this assumption,} we show that any work-conserving scheme is delay-optimal {among all on-line scheduling schemes} when $k=1$. When $k>1$, we find that a simple greedy scheme, which allocates all available threads to the head of line request, is delay optimal {among all on-line scheduling schemes. We also provide some numerical results that point to the limitations of the exponential assumption, and suggest further research directions.}
\end{abstract}

 \section{INTRODUCTION}\label{intro}
Cloud storage, an essential element of cloud computing, has been rapidly expanding, backed by technology giants, such as Amazon, Google, Microsoft, IBM, and Apple, as well as now popular startups, such as Dropbox, Rack-space, and NexGen. The total cloud storage market is expected to grow from \$5.6 billion in 2012 to \$46.8 billion by 2018 with a compound annual growth rate  of 40.2\% \cite{CloudStorageReport}.  Cloud storage systems provide users with easy access, low maintenance, flexibility, and scalability.

\begin{figure}[!t]
\begin{center}
\includegraphics[width = 0.9\linewidth]{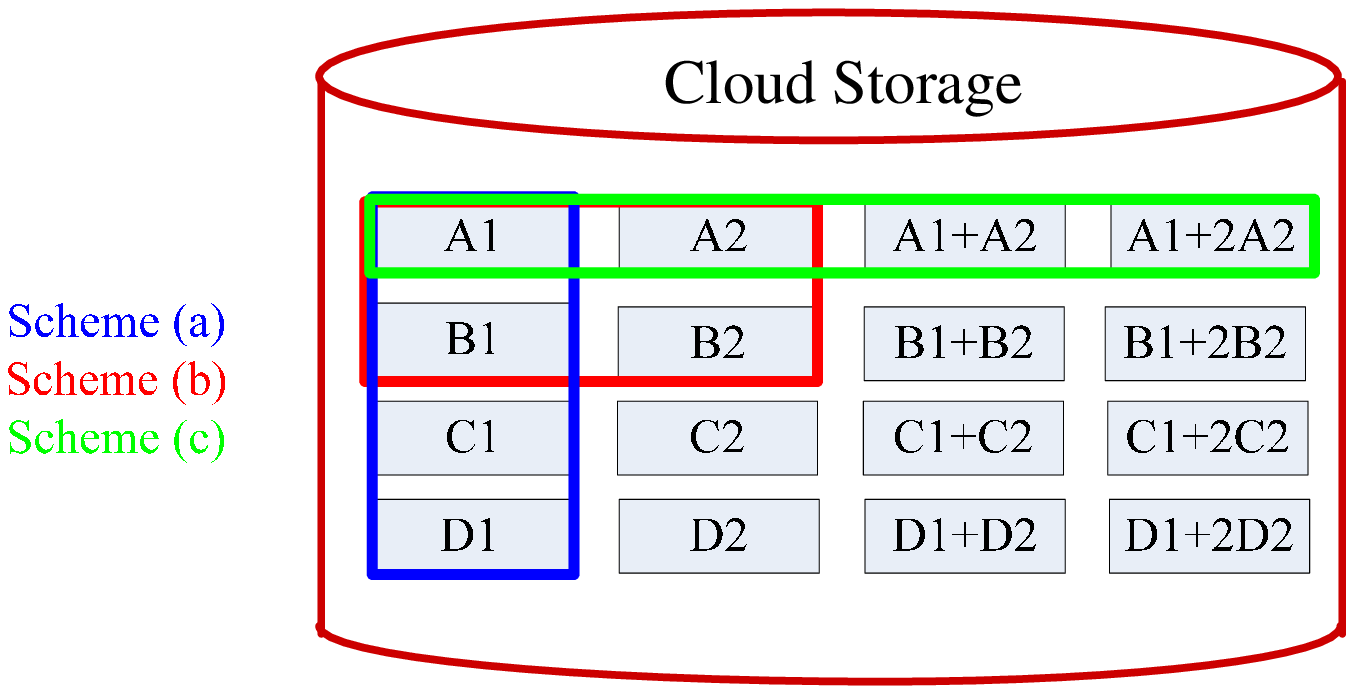}
\caption{An example of four files using  FEC codes with parameters (4,2) in cloud storage and three thread allocation schemes}
\label{fig:mdscode}
\end{center}
\end{figure}

\IEEEpubidadjcol

 As in all storage systems, efficiency, reliability, and latency are critical requirements for cloud storage systems. Although simple duplications can be used for reliability, the most promising storage codes are forward-error-correction (FEC) codes, such as maximum-distance-separable (MDS) codes, which provide
a better resiliency to erasures than duplication for a given amount of redundancy. In this paper, we consider FEC codes with fixed coding parameters $(n,k)$, i.e., each file is divided into $k$ equal-sized chunks, which are then expanded into $n$ chunks such that any $k$ out of $n$ chunks are sufficient to successfully restore the $k$ original chunks (hence, the file itself). Fig. \ref{fig:mdscode} shows an example, where $n=4$ and $k=2$. In particular, four equal-length files A, B, C, and D are stored in the cloud using FEC codes with parameters (4,2). Each file $x\in \{A, B, C,D\}$ is partitioned into two equal-length chunks, $[x_1,x_2]$, and its  four coded chunks, $[x_1$, $x_2$, $x_1+x_2$,  $x_1+2x_2$], are stored in the cloud. A file can be retrieved from any two of its four chunks. In this manner, different levels of efficiency and reliability can be achieved by adjusting the parameters $(n,k)$ of FEC,  where  we have $k>1$ in general, while  $k=1$ simply means data duplication.   Some online
file storage companies have already adopted such FEC codes, such as Wuala \cite{wuala}.

While reliability and efficiency have been carefully studied in cloud storage systems, its  performance in terms of delay has received much less attention, even though delay is a critically important issue that significantly affects user experience and can play a major role in the widespread adoption of these systems.
 Measurement studies show that there exists a significant skew in network bound I/O performance  \cite{addShengboOriginal}. Earlier evaluations of Amazon S3 indicate that the slowest 10-20\% of read/write requests see more than 5$\times$ of the mean delays (e.g., see \cite{Garfinkel07anevaluation}). The experimental study of Amazon S3 in \cite{addShengboOriginal} shows a similar trend. In particular, the average delay of reading a 1MB file is 139\textit{msec}, with 80\% delay as 179\textit{msec}, 95\% as 303\textit{msec}, 99.9\% as 811\textit{msec}, respectively.
It has been shown that delay tail has a large impact on user experience and service provider revenue, e.g.,  every 500 ms extra delay in service will lead to a 1.2\% user loss for Google and Amazon \cite{lu-jiq-2011}. Thus, it is  important to reduce the delay spread, i.e., cutting the long tail of the service time.

\begin{figure}[!t]
\begin{center}
\includegraphics[width = 0.95\linewidth]{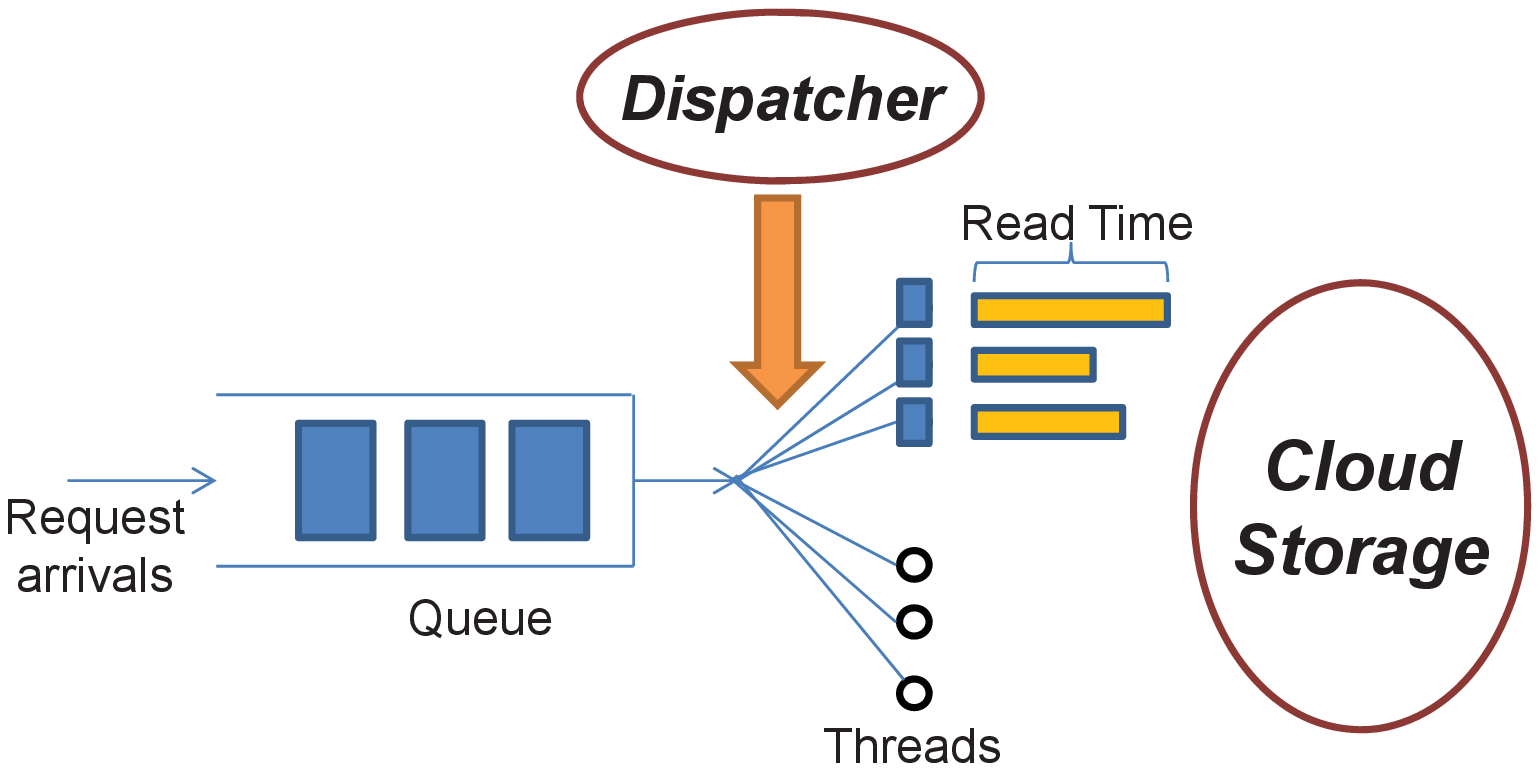}
\caption{The architecture for data retrieval through a multiple-thread dispatcher.}
\label{fig:arch}
\end{center}
\end{figure}

%
% \begin{figure}[!t]
% \begin{center}
% \includegraphics[width = 0.8\linewidth]{threeschemes1.eps}
% \caption{Three thread allocation policies.}
% \label{fig:thread}
% \end{center}
% \end{figure}

 It has been revealed  that the delay exhibits weak correlations between different read tasks \cite{addShengboOriginal,Garfinkel07anevaluation}. Thus, the delay can be reduced by simply transmitting data replications in parallel.
This motivates us to investigate the critical issue of improving the system latency of data retrieval. In particular, a key question we ask is: \emph{can we leverage the inherent redundancy provided by FEC for reliability to improve delay performance?}
   Assume that a user requests to download  files, i.e., their encoded chunks, from the cloud storage systems with a dispatcher. The dispatcher  schedules these downloading requests with a number of threads, where each thread can be used to retrieve one chunk each time. The number of threads that are assigned to the user is based on the function of the service that the user is willing to pay for --- the higher the payment, the greater the number of threads allocated to the user. The model is depicted in Fig. \ref{fig:arch}.
A typical application scenario is given as
follows: an online game user needs to load and save (read and write) his data into the clouds using FEC codes. The user is allocated with a number of threads, where the number of threads is determined by how much the user is willing to spend.

In this paper, we  focus on the delay of read requests (Note that our results can
be also applied to write requests).    We study the read delay induced by different thread allocation
schemes, i.e., how the threads are assigned to  retrieve distinct file chunks.
 To make the problem more concrete, consider the four encoded files A, B, C, and D, shown in Fig. \ref{fig:mdscode}. An user  requesting
these four files  has four threads and can allocate these  four  threads using different allocation
schemes, as shown in Fig.~\ref{fig:mdscode}. In scheme (a),  one thread is allocated to each file, i.e., four threads are  allocated to download A1, B1, C1 and D1, depicted by the blue box. After  downloading these chunks, the threads
are then scheduled to download the second chunks of all four files.
Scheme (b), depicted by
the red box, provides  parallelism where the two chunks of file A/B, i.e., A1, A2, B1 and B2, are retrieved in parallel. Scheme (c) further exploits the parallelism, and allocates all four threads to file A, each requesting a chunk, i.e., A1, A2, A1+A2 and A1+2A2, as depicted by
the green box. As long as any two of the four  chunks  are retrieved, the dispatcher immediately \textit{terminates} the other threads  and allocates them to the next file. Clearly,  the service delay of file A is minimized by scheme (c) since the file can be constructed from the first two successful chunks. On the other hand,  scheme (a), compared to the previous two schemes, allows  a greater
number of parallel requests that can be served, i.e., four for scheme (a), two for scheme (b)
and one for scheme (c). Therefore, there exists a tradeoff between the service time of a file and parallelism in data retrieval. And it is interesting to ask which scheme has the best delay performance.

% Although scheme
% (c) may introduce additional cost as more threads are used for each file. Therefore, It is imperative to design FEC and thread
% scheduling policies that well balance the delay performance and the cost.

In this paper, we resort to queueing theory and notice that the model can be depicted as a multiple-server queue  with arrivals of data retrieving requests and a server corresponding to a thread. However, this is not a typical queueing model because a server can terminate its operation, depending on when other servers complete their service (due to the redundancy that is spread across the threads).
{Based on observations made in \cite{addShengboOriginal}, on real traces that are measured over Amazon
S3 (see Section \ref{measurement}), the time to retrieve a fixed size chunk is random and can be approximated as a constant delay plus an \textit{i.i.d.} exponentially distributed random variable. The randomness
is because the communication within the cloud causes random delays.}
% To the best of our knowledge, this is  an uncharted area, with the exception of very   recent work \cite{addShengboOriginal}, \cite{Joshi-2012,huang-isit-2012} that  studies the performance of  scheme (b).
% A somewhat surprising result of the paper is that Scheme (c), or more general, the greedy policy is indeed delay-optimal %(i.e., using all 4 threads to retrieve 2 chunks)
% under exponentially distributed service time! In addition, scheme (c) not only reduces average delay, but also delay spread, a much desired result in cloud storage. \xl{please check these statements are indeed accurate and correct.:) }

%We instantiated an extra large EC2 instance with high I/O
%capability in the same availability region as the S3 bucket that
%stores our objects. We run experiments within North California
%as well as Tokyo regions.
%
% necessitating solutions that provide robustness in a cost effective
%manner
%

In this paper, we make the following contributions:
\begin{itemize}
\item
%Shengbo, is the architecture really new? It seems to me that this has been studied before, while the analysis and optimality aspects are new.
We propose a queueing architecture that leverages the coding redundancy inherent in cloud storage systems, to improve file retrieve latency performance. In particular, we present a new queueing model to study data retrieve latency,  where each file is encoded and stored in the cloud using FEC codes.
%\footnote{We refer to Section \ref{system} for formal definitions of $n$ and $k$.}.

\item
{Under the simplifying assumption that the chunk downloading time is \textit{i.i.d.} exponentially distributed, we analyze the delay performance of different scheduling schemes.} %We analyze the delay performance under several schemes assuming that  the service time to download a chunk is an exponentially distributed random variable.
When $k=1$, we show that any work-conserving scheme that fully utilizes all available threads is delay-optimal {among all on-line scheduling schemes}.
For the case when $k>1$, we  prove that a simple greedy policy is delay-optimal {among all on-line scheduling schemes}.
This is a somewhat surprising result, as delay optimality is rather strong
in general, which also implies throughput optimality.

%\item We validate our theoretical results using both exponentially distributed service time distribution and real traces. We show that making use of inherent redundancy introduced by FEC can indeed reduce the delay and significantly reduce the delay spread.
\end{itemize}

%To the extent of our knowledge, we are the first to prove the delay-optimal scheme under the exponentially distributed downloading time in cloud storages system deploying FEC.

The organization of the paper is as follows. We discuss related work in
Section \ref{related}. In Section \ref{measurement}, we present our measurement results over Amazon S3. Section  \ref{system}
describes  the system model.
In Sections \ref{kequalone} and \ref{kgreaterone}, we  present the results
under the cases of $k=1$ and $k>1$, respectively. The simulation results are presented in  Section \ref{simulation}.  We then conclude our paper in Section \ref{conclusion}.

% \xl{My notes}
%
%  We note that another advantage of Scheme (c), compared to Scheme (a) and (b), is its simplicity.
%
%
% \xl{I remove the write part. I share the similar concern with Infocom reviewer 1. Improving reading delay is good enough a motivation.}
% \xl{emphasize cut tail, i.e., delay spread}

  \section{RELATED WORK}\label{related}
The use of coding to improve the efficiency of large-scale data storage systems has received significant attention  \cite{dimakis-2010}. For instance, the authors of \cite{dimakis-2009} used interference alignment for reducing repair traffic for storage systems with erasure coding; in \cite{rashmi-2010}, the authors constructed explicit regenerating codes for achieving minimum bandwidth in distributed  storage systems; in \cite{leong-2009}, the authors considered the problem of allocating capacity for optimal data storage.
However, most of the existing works have focused mainly on utilizing coding to reduce repair bandwidth and storage capacity for storage systems (for more discussion, please see \cite{coding-survey-dimakis} and the references therein).
%Shengbo, above here you can't leave the sentence hanging, you need to emphasize why this is different from your work.
%

Recently, due to the increasing importance of service latency (i.e., delay) as a system performance metric, e.g., for Google and Amazon, every $500$ ms extra delay in service will lead to a $1.2\%$ user loss \cite{lu-jiq-2011}, researchers have started to study the effect of coding on content retrieval delay for data storage systems.
The work \cite{alex-delay-linenet09} considered delivering a set of packets over a linear network with minimum delay. In \cite{ferner-2012}, coding was used to reduce the blocking probability in storage networks.
In \cite{ali-poe-2012}, coding was used as a way to prevent service interruption. In \cite{RRighter} and \cite{RRighter1}, the authors investigated the performance of using multiple servers to download file replications without considering coding. The authors in \cite{addShengboOriginal} proposed a heuristic transmission control scheme by dynamically adjusting the coding parameters which demonstrates good delay performance in cloud storages. 

The most related theoretical results that we know of are \cite{huang-isit-2012,shah-mdsq-2012,Joshi-2012}, which investigated how to assign the storage disks to serve the read requests for reducing delay. In \cite{huang-isit-2012}, all the requests are put in a centralized queue and the authors showed that FEC codes can reduce the data retrieving delay compared to simple data replications. In \cite{shah-mdsq-2012}, the authors proved that flooding requests to all storage disks, rather than a subset of disks, has a shorter data retrieving delay for the centralized queueing model. In \cite{Joshi-2012}, the requests are dispatched to multiple local queues at the storage disks and it was showed that coding reduces data retrieving time. Different from \cite{huang-isit-2012,shah-mdsq-2012,Joshi-2012}, we focus on the systems with limited downloading threads (bandwidth) and study how to allocate the downloading threads to exploit the storage redundancy and minimize the data retrieving delay. Since our system model is quite different from the existing studies, novel proof techniques are employed.

 \section{Measurement Over Amazon S3}\label{measurement}

\begin{figure}[!t]
\centering
\includegraphics[width=0.8\columnwidth]{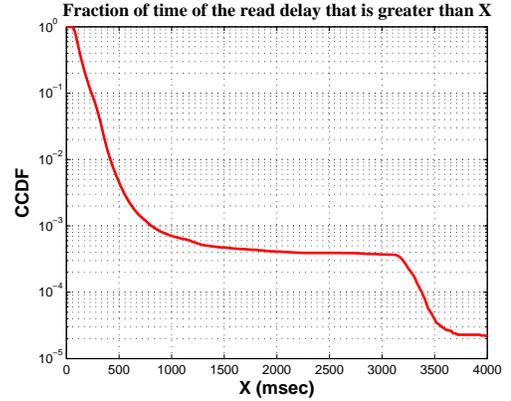}
\caption{CCDF of read delay for 1MB chunk.}
\label{fig_1_2}
\end{figure}
\begin{figure}[t]
\centering
\includegraphics[width=0.7\columnwidth]{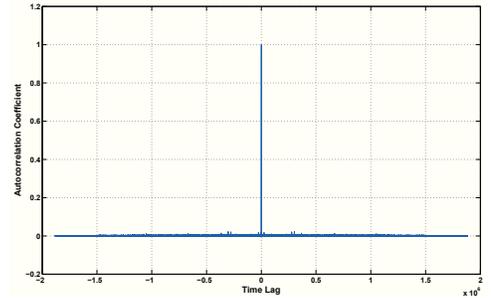}
\caption{ Correlation between time ordered delay samples.}
\label{fig_1_3}
\end{figure}

In  this section,  we  describe some measurement results  of the read delay using Amazon S3 made by our coauthors from Docomo Innovations, Inc. \cite{addShengboOriginal}.
 Please refer to \cite{addShengboOriginal} for more
details.

Fig. \ref{fig_1_2} plots the complementary cumulative distribution function (CCDF) of the delay for downloading a file of size 1MB. We can see that the downloading time indeed observes a wide spread, although the file size is the same across all experiments.  Another observation is that despite the delay floors observed at very low percentiles (e.g., beyond the 99.9th percentile), up to
99th percentile, the CCDF is close to a linear term in delay and note that the y-axis is logarithmic. {This indicates that the chunk downloading time can be approximated as a constant delay plus an \textit{i.i.d.} exponentially distributed random variable.}

%an exponential distribution is a good approximation of the service time for most of the distribution, as proposed in the model.

Fig. \ref{fig_1_3} shows the autocorrelation coefficient between the service times of consecutive read tasks. Note that the mean delay is subtracted from the delay samples. We can see that there is negligible correlation between consecutive download delays. % Based on these observations,  we  approximate the downloading times for any thread as \textit{i.i.d.}  exponentially distributed random variable.

%Shengbo, I've only made minor changes to this section. Please read and see what you think.

 %&latex

\section{SYSTEM MODEL}\label{system}

%\begin{figure}[!t]
%\centering
%\includegraphics[width=0.8\columnwidth]{Visio-sys.eps}
%\caption{System Model}
%\label{fig_system}
%\end{figure}

%Shengbo, here it would be a good idea to connect up with the example that you plan on introducing  in the Introduction.
As mentioned in the introduction, we consider a  cloud storage system, where data is stored using  FEC codes with coding parameters $(n,k)$, i.e., each
file has $n$ chunks stored in the clouds and any  $k$ out of the  $n$ chunks are sufficient to successfully restore the  original file. We assume that  the files in the cloud are  homogeneous with same size and coding parameters,
 which also implies that all chunks are homogeneous with same size.
%The case of different file size and coding parameters will discussed later.

Data retrieval (read) is provided through a dispatcher with multiple threads,
 as shown in Fig. \ref{fig:arch}.  When read requests arrive, they are first enqueued. The dispatcher  then determines how to allocate   threads  for each request. For instance, in Fig.~\ref{fig:arch}, there are three threads (rectangles) which are scheduled to download chunks
from the cloud, while the other three threads (circles) are idle. Due to the service chosen by the  user, the dispatcher can simultaneously keep at most $L$ threads active in any time instant.
%Shengbo, in the above sentence do you want to put this down to resource limitation, or should you say something like: Due to the service chosen by the end user, the proxy server can simultaneously keep at most $L$ threads active in any time instant.
We assume that the read requests arrive at the dispatcher  with rate $\lambda$. %, and are first buffered in a FIFO (first-in-first-out) queue.
Then, when a file is to be read from the cloud, the dispatcher schedules $m$ ($k\leq m\leq n$) read tasks for distinct chunks of the file by activating a total of $m$ threads (not necessarily simultaneously). Due to the use of FEC codes, the first $k$ successful responses from the storage cloud will then be sufficient for completing the read operation. Hence, we assume that once $k$ chunks are successfully downloaded, the other read tasks in progress can be terminated immediately.

For simplicity, we first assume  $n \geq L+k-1$. This assumption means that there are always enough distinct data chunks in the clouds, since it guarantees that there are still at least $L$ distinct chunks in the cloud
for the $L$ threads to share given that $k-1$ chunks have been successfully downloaded.
 In this paper, we investigate the file retrieve latency under different thread allocation schemes, i.e., how should the dispatcher schedule the threads for serving the read requests.
 %

%whenever there is any available thread and the queue is not empty.
%\textit{More specifically, the problem is which read request  the dispatcher schedules the available threads for and when the dispatcher schedules the available threads?
%}

% \begin{enumerate}
%
% \item \textbf{Post-Dispatcher}: The dispatcher makes decisions before the
%  request enters the threads and
%
% %  After a thread  finishes its assigned
% % read/write operation for a request, it is not allowed to serve the same request again.
%
%
%  \item \textbf{Pre-Dispatcher}: The dispatcher makes decisions before the
%  request enters the threads and
% %
% % After a thread  finishes its assigned
% % read/write operation for a request, it is allowed  to launch a new read/write operation for the same request given that the request has not departed.
% \end{enumerate}

%  We can try to solve
% this problem using several steps:
% \begin{enumerate}
% \item $k=1$, and the CDF is exponential distributed.
% \item $k>1$, and the CDF is exponential distributed.
% \item $k=1$, and the CDF is generally distributed.
% \item $k>1$, and the CDF is generally distributed.
% \end{enumerate}

%&latex
\subsection{Problem Formulation}\label{notations}
%In this section, we specify our system model with our notations.
We consider a time period $[0, T]$ and let $\mathcal{N}=\{1,2,3,$ $\cdots,i,\cdots
N_T\}$ denote the set of read request arrivals during this period, where  $i$ denotes the $i$-th arrival
and $N_T$ is the total
number of request arrivals during the period.
We denote $n_i$ as the  number of threads that the dispatcher allocates for the request $i$. Clearly, $k\leq n_i\leq n$.  We denote the downloading time of the $j$th thread for request $i$ by $X_{i,j}$, $j=1, \cdots,n_{i}$.
We let $T_{A}^{i}$ denote the arrival time of request $i$, and  $T_{S}^{i,j}$ denote the starting time of the $j$th thread of  request $i$. Without
loss of generality, we assume $T_{S}^{i,1}\leq T_{S}^{i,2}\leq\cdots\leq
T_{S}^{i,n_{i}}$. We denote the  finishing time of the thread $j$ of the request $i$ as $T_F^{i,j}$, which is given by $T_F^{i,j}=T_{S}^{i,j}+X_{i,j}$. Note that the threads are  ordered by their starting time but not the completion time. So it is possible that $T_F^{i,j} > T_F^{i,l}$ even if $j<l$.
The departure time  of request $i$, denoted as $T_{F}^{i}$, is then given by the time when $k$ of its threads have finished. Let $T_F^{i,1:n_{i}}\le T_F^{i,2:n_{i}}\le \cdots\le T_F^{i,n_{i}:n_{i}}$ be the sorted permutation of the finish times of request $i$'s threads. Thus, we have $T_F^i = T_F^{i,k:n_{i}}$.

The {system} delay for request $i$, denoted as $D_i$ is therefore given by
\begin{align}
\label{eq_1}
D_i=T_F^i-T_A^i.
\end{align}

Hence, the total expected delay under a thread allocation policy $\pi$ is given by
\begin{align}
\label{eq_2}
\mathbb{E}_{\pi}[D]=\lim_{T\rightarrow\infty}\frac{1}{N_T}\sum_{i=1}^{N_T} \mathbb{E}_{\pi} [D_i],
\end{align}
where the expectation is with respect to the distributions of   arrival
process and departure process.

In the following, we will develop a new queueing model for analyzing such cloud systems that use FEC codes for data storage. Before we proceed, we first have the following definitions.

\begin{definition}\label{def_1}
The queue is said to be stable under a policy $\pi$ if
\begin{align}
\label{eq_stability}
\mathbb{E}_{\pi}[D] < \infty.
\end{align}
\end{definition}

\begin{definition}\label{def_2}  \textit{Capacity Region $\Lambda$:} The set  of request
arrival rates under which the queue can be stabilized by some possible scheme.
\end{definition}

\begin{definition}\label{def_3}  \textit{Throughput-optimal scheme:} A scheme
is said to be throughput optimal  if the queue with arbitrary arrival rate $\lambda$ can be stabilized under this scheme whenever there exists an $\epsilon>0$ such that $\lambda+\epsilon\in\Lambda$.
% the request arrival rate is in the capacity region.
\end{definition}

\begin{definition}\label{def_4}  \textit{Delay-optimal scheme:} A scheme
$\pi$
is said to be delay optimal  if  it yields the smallest $\mathbb{E}_{\pi}[D]$
among all schemes for any  request arrival rate $\lambda$ such that $\lambda+\epsilon\in\Lambda$ for some $\epsilon>0$.
%in the capacity region.
\end{definition}

  %&latex
%

\subsection{Thread Allocation Schemes}\label{anewqueueingmodel}

In this section, we first present several thread allocation schemes that can be adopted by the dispatcher and then we motivate our problem using a
simple example. The model can be depicted as a multiple-server queue  with arrivals of data retrieving requests and a server corresponding to a thread. However, this is not a typical queueing model because a server can terminate its operation, depending on when other servers complete their service (due to the redundancy that is spread across the threads).
We here consider two  schemes: the \textsf{greedy} scheme and the \textsf{sharing} scheme.
\begin{enumerate}
\item \textsf{The greedy scheme}: All  $L$ threads are always allocated to the HoL (head-of-line) request simultaneously until it departs. During the serving process, if any thread out of $L$ finishes downloading its assigned chunk, it  immediately starts to download another distinct chunk belonging to the \textit{same} file. Then, at some point, one thread
finishes  downloading so that
 the cumulative number of successfully downloaded chunks reaches $k$, all the other $L-1$ threads in progress  get terminated immediately. The read request is considered  complete and departs the queue. After that, all the threads will  be allocated to serve the next HoL read request if the queue is not empty. Otherwise, all threads remain idle.

% \item \textsf{The sharing scheme}: At any time instance when some threads become idle, all the available threads will be equally shared among all requests in the queue. \textcolor{red}{Shengbo, what happen if we have $K$ free threads but there are no more new read request in the queue?.}

\item \textsf{The sharing scheme}: The dispatcher always allocates exactly $k$ threads to each request (not necessarily simultaneously). The requests
are served in a first-come-first-serve manner.  The dispatcher allocates as many available threads as possible to a request until $k$ threads have been assigned to it in total. When the number
of allocated threads for a request reaches $k$, the dispatcher will allocate
the available threads to the next request.
\end{enumerate}

We can see that the greedy scheme and the sharing scheme are  two extremes.
The greedy scheme allocates the maximum possible resources to each individual request, but it can  serve only one request at any time. On the other hand,   the sharing scheme is the most conservative for each individual request,
  yet  serves the maximum possible parallel requests. Another observation  is that the total number of chunks required for the greedy scheme
is $L+k-1$, which consists of $k$ successfully downloaded and $L-1$ terminated. This corresponds to the assumption $n\geq L+k-1$.

To better motivate our problem, we consider a simple example  as shown in Fig.  \ref{fig_example}.  Two threads are used to read two files A and B, which use coding parameters $(2,1)$, i.e.,    each  file simply has two duplications. We compare the delay of the two files using the  two schemes under two different downloading time distributions.

\begin{figure}[!t]
\centering
\includegraphics[width=0.8\columnwidth]{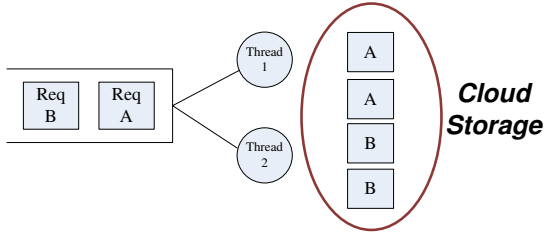}
\caption{An example of a two-thread queue}
\label{fig_example}
\end{figure}

\textbf{Case 1) }If the downloading time is constant for any chunk, we can see that the sharing scheme always outperforms the greedy scheme in
terms of the delay performance since the parallel downloads do not bring any benefit when the delay is a fixed value.

\textbf{Case 2) }Consider a distribution of downloading time, which is 0 with probability
2/3 and 3000\textit{ms} with probability 1/3. Therefore, we can see that sharing scheme
has an expected delay of 1000\textit{ms} for each request, since it allocates one request
to each thread. On the other hand, since the greedy scheme allocates both threads
to a request one by one, we can easily check that the expected delay for the first request is 1000/3 \textit{ms}
and the expected delay for the second one is 2000/3 \textit{ms}.  Both experience smaller delay than those under the sharing scheme.

The detailed calculation is as follows:
\begin{enumerate}
\item For the first request, it has a delay of 3000\textit{ms} if and
only if both threads suffer a delay of 3000\textit{ms}. The probability of
this scenario is  1/9. With a probability of 8/9, the delay is 0. Therefore,
the expected delay for the first request is $\frac{3000}{9}+\frac{0\times8}{9}=1000/3$
\textit{ms}.
\item For the second request, the expected delay is  the sum of the expected
waiting delay  in the queue and its expected service delay.
Notice that its expected waiting delay in the queue is exactly the expected
delay of the first request, i.e., 1000/3 \textit{ms}. In addition, its service delay can be calculated in the same way as the first request, that is, 1000/3 \textit{ms}. Thus,
the expected delay for the second request is 2000/3 \textit{ms}.
\end{enumerate}

 Hence, we can observe that the delay-optimal scheme heavily depends on the distribution of the downloading time.
This motivates us to ask a key question: \textit{what is the delay optimal thread allocation
scheme when the downloading time exhibits a  distribution like
the one in Fig. \ref{fig_1_2}}. {For the tractability of theoretical analysis, we assume that $X_{i,j}$ is \textit{i.i.d.}~exponentially distributed random variable with mean $\mu$ in the following sections.}

 %In the following sections, as we discussed in Section \ref{measurement}, we approximate the downloading time for any individual thread as an \textit{i.i.d.}   exponentially distributed random variable with parameter $\mu$.

%

\section{ANALYSIS of CASE $\lowercase{K}=1$}\label{kequalone}
To facilitate the understanding of our results, we start by analyzing the case when $k=1$. In this case,  it simply means using data duplication to store the files.
To present our analysis, we first define the following:
\begin{definition}\label{def_5}
\textit{Work-conserving schemes:}  A scheme is said to be work-conserving if no thread is idle whenever there are  requests waiting in the queue.
\end{definition}

We emphasize that \textbf{a throughput-optimal scheme
must be work-conserving}. In addition, \textbf{a delay-optimal scheme must be work-conserving}. The reason is that the delay and throughput performance of any non-work-conserving scheme can be improved by assigning the idle threads to download some additional chunks.

%When a thread becomes idle and the queue is not empty, a
%work-conserving scheme will schedule the thread for some
%request in the queue immediately (not necessarily the HoL
%request). For instance, when there are more than one request in
%the queue, a work-conserving scheme may allocate idle threads
%to serve the second request in the queue. Thus, there are potentially many different work-conserving schemes. However, it is worth noting
%that while the greedy scheme is work-conserving, the sharing
%scheme is not. For example, when there is only one request in the queue and
%there are more than $k$ idle threads, we can see that some threads will still
%be idle according to the definition of the sharing scheme.  In addition, for some work-conserving
%schemes, it may be wasteful because some threads  have useless (unfinished) downloads due to thread  terminations.

% \begin{figure}[!t]
% \centering
% \includegraphics[width=0.8\columnwidth]{Visio-mc.eps}
% \caption{State Transition Graph}
% \label{fig_mc}
% \end{figure}

Before we state our results, a key property of work-conserving scheduling schemes for both $k=1$ and $k>1$ cases is described as follows:

\begin{lemma}\label{lem_1}
If the downloading time of each individual thread is \textit{i.i.d.} exponentially distributed with rate $\mu$, and all $L$ threads are active, then the service time for the threads to download one more coded chunk is exponentially distributed with rate $L\mu$ and is \textit{i.i.d.} across coded chunks.
\end{lemma}

\begin{proof}
Suppose that the system starts to download the next coded chunk at time $t$, either because a new request arrives or because a coded chunk is downloaded at $t$. The downloading operation of Thread $l$ may start before time $t$. Let $R_l$ be the resident downloading time of Thread $l$ after time $t$. Since all $L$ threads are active, the service time for the threads to download one more coded chunk is given by
\begin{equation}
R =\min_{l\in\{1,\cdots,L\}}R_{l}.
\end{equation}
According to the memoryless property of exponential distribution, the $R_{l}$'s are also exponentially distributed with rate $\mu$ and are \emph{i.i.d.} across threads. Therefore, $R =\min_{l\in\{1,\cdots,L\}}R_{l}$
is exponentially distributed with rate $L\mu$. By the memoryless property of exponential distribution, the download durations of the retrieved coded chunks are \textit{i.i.d.} Therefore, the asserted statement is proved.
\end{proof}

Now, for $k=1$ case,  we have the following result.

\begin{theorem} \label{theorem_1} When $k=1$ and $n \geq L$, given that the downloading time of each individual
thread is \textit{i.i.d.} exponentially distributed, any work-conserving scheme is throughput optimal and also delay optimal {among all on-line scheduling schemes} for any arrival process.
\end{theorem}
\begin{proof} First, we  show that a delay optimal scheme must also be  throughput optimal. Suppose this is not true and denote a delay optimal scheme in consideration as $\pi^*$, which is not throughput optimal. Then, according to Definitions
\ref{def_1} and \ref{def_3},  there exists an arrival rate $\lambda$ such that $\lambda+\epsilon\in\Lambda$, but under this rate, $\pi^*$ results in $\mathbb{E}_{\pi^*}[D]=\infty$. However, since $\lambda+\epsilon\in\Lambda$,  there exists a policy $\pi$ that has $\mathbb{E}_{\pi}[D]<\infty$, and hence $\mathbb{E}_{\pi}[D]<\mathbb{E}_{\pi^*}[D]$. This  contradicts the delay optimality  of $\pi^*$.

%Suppose there exists a delay optimal scheme that is not throughput optimal.
%It means that there exists some arrival rate in the capacity region such that the delay optimal scheme cannot yield a stabilized queue under this
%arrival rate. However, since
%this arrival rate is within the capacity region, there must be some scheme that can stabilize the queue,
%which  implies a smaller delay. It contradicts with the fact the delay optimal
%scheme has the smallest expected  delay among all schemes. Hence,
%we have  that delay optimal scheme is always  throughput optimal. Therefore,
%it suffice to show any work-conserving scheme is delay optimal.

Now we prove the delay optimality part. We only need to consider the work-conserving schemes, because a non-work-conserving scheme cannot be delay optimal. In particular, we can always transform a non-work-conserving scheme to a work-conserving scheme by utilizing the idle threads to download some additional chunks, which leads to a lower delay.

Since $k=1$, each downloaded coded chunk leads to a request departure. %When the queue is not empty, any work-conserving scheme keeps all the threads ON according to Definition \ref{def_5}. Since the downloading time for any
According to Lemma \ref{lem_1}, the service durations of the downloaded coded chunks are \textit{i.i.d.} exponentially distributed under any work-conserving scheme. Therefore, the service durations of the requests are also \textit{i.i.d.} exponentially distributed. Different work-conserving schemes only affect the service order of different requests, but have no influence on the average delay. Therefore, any work-conserving scheme is delay optimal. By this, the asserted statement is proved.
%individual thread  is \textit{i.i.d.}  exponentially distributed, together
%with the memoryless property of exponential
%distribution, the potential departure process
%of effective chunks, i.e., the departure process  assuming the queue is always backlogged,  from these $L$ threads is
%a Poisson process with rate $L\mu$.
%  Given the same arrival process, the  distribution of actual departure process only
%depends on the distribution of the potential departure process. Hence, the actual departure processes under different work-conserving schemes, have the same distribution, and thus have the same rate. In addition,  we can always transform a non-work-conserving scheme to a work-conserving scheme by utilizing all idle threads when the
%queue is not empty,  which leads to a larger departure rate. Thus, any work-conserving
%schemes   guarantees a maximal departure rate among all schemes.
%Hence, we conclude that any work-conserving scheme is delay optimal and thus throughput optimal.

%
% Note that the inter-arrival time and the inter-departure time are both exponentially distributed. In particular, the
% system can be modeled by a continuous-time Markov chain (CTMC), with the state being the number of requests in the system.
\end{proof}

%\textbf{\textit{Remark:}}
It is interesting to see that although some work-conserving schemes, such as the \textsf{greedy} scheme, appears to ``waste'' some system resources because some threads have useless (unfinished) downloads due to redundant assignment, it is still both throughput and delay optimal. The reason is because the chunk downloading time is assumed to be exponentially distributed, which satisfies the memoryless property.

% abandoned threads are the ones suffering long delay.
%The service time for any particular file is reduced compared to the non-wasted
%schemes.
%We will still observe such ``wastage but optimality'' in the next section.

%
% \begin{corollary} For repeat-forbidden schemes, greedy scheme is  delay optimal.
%
% \end{theorem}
% \begin{proof} Could be derived from the previous theorem.
% \end{proof}

%&latex
\section{ANALYSIS of CASE $\lowercase{K}>1$}\label{kgreaterone}

In this section, we extend our results to the case when $k>1$. We first have the following definitions.

\begin{definition}\label{def_6} \textit{Effective chunks:} chunks that are downloaded by the first $k$ completed threads for any particular file.
\end{definition}

\begin{definition}\label{def_7}  \textit{Thread terminations:}  chunk download attempts that are terminated due to the completion of the read request, i.e., there are already $k$ effective chunks downloaded.
\end{definition}

% We   first prove the following lemma.
%
%
%\begin{lemma}  \label {lemma_0}
%The departure process of effective chunks from these $L$ threads is a Poisson process, with rate $L\mu$, when the system is always backlogged, i.e., the queue is always non-empty.
%\end{lemma}
%
%\begin{proof}
%At any time, if the system is always backlogged, all threads are ON due to  Definition \ref{def_5}. Since the downloading time for any
%individual thread is \textit{i.i.d.}  exponentially distributed, together
%with the memoryless property of exponential
%distribution, the departure process of effective chunks from these $L$ threads is a
% Poisson
%process with rate $L\mu$.
%\end{proof}

Next, we  have the following theorem regarding the throughput optimality of the work-conserving scheduling policies.

\begin{theorem} \label{theorem_2} When $k>1$ and $n \geq L+k-1$, given that the downloading time of each individual
thread is \textit{i.i.d.}   exponentially distributed, any work-conserving scheme  is throughput optimal {among all on-line scheduling schemes for any arrival process}.
\end{theorem}

\begin{proof} We only need to consider the work-conserving schemes, because a non-work-conserving scheme cannot be throughput optimal.
% when  %with the assumption that the system is always backlogged, i.e., the queue is always non-empty. %Hence, according to the Definition \ref{def_5}, we know that  all the threads are always on without idle periods.
According to Lemma \ref{lem_1}, we know that the service durations of effective chunks under any work-conserving scheme have the same distribution, i.e., they are  \textit{i.i.d.} exponentially distributed with rate $L\mu$. Since each request requires to download
exactly $k$ effective chunks, the average request departure rate of any work-conserving scheme is $L\mu/k$. On the other hand, the queue is stable if and only if the average request arrival rate is less than the average request departure rate. From the previous discussion,
any work-conserving scheme can provide the maximum request departure rate $L\mu/k$. Therefore, any work-conserving scheme  is throughput optimal.
\end{proof}

We now have the following theorem regarding the delay performance of the \textsf{greedy}  scheme.

\begin{theorem}  \label {theorem_3}
When $k>1$ and $n \geq L+k-1$, given that the downloading time of each individual
thread is \textit{i.i.d.} exponentially distributed, the \textsf{greedy} scheme is delay optimal {among all on-line scheduling schemes} for any arrival process.
\end{theorem}

\begin{proof}
First, notice that a delay-optimal scheme must be work-conserving.
Otherwise, it is easy to reduce the delay by simply allocating the idle threads to download more chunks.

Let $s_i$ denote the arrival instants of the $i$-th arrival effective chunk and $t_i$ denote the departure instants of the $i$-th departed effective chunk. Obviously, we have $s_i<s_{i+1}$, $t_i<t_{i+1}$, and $s_i<t_{i}$.
Fix the arrival process $\omega_A=\{s_1, s_2, \cdots\}$ of the effective chunks, we will show that the distribution (probability density function) of the departure process $\omega_D=\{t_1, t_2, \cdots\}$ remains the same for any work-conserving scheme.

According to Lemma \ref{lem_1}, the service time of an effective chunk has the same distribution under any work-conserving scheme. Let $S_i$ denote the service time of the $i$-th departed effective chunk.
Given the arrival instant $s_1$ of the first effective chunk, the departure instant $t_1$ is given by $t_1 = s_1+S_1$. Moreover, since the distribution of $S_i$ remains the same under any work-conserving scheme, the distribution of $t_1$ also remains the same under any work-conserving scheme.
Now consider the departure instant $t_i$ of the $i$-th departed effective chunk. Since the system is work-conserving, $t_i$ is given by $t_i= \max\{s_i,t_{i-1}\} + S_i$. Since the distributions of $\{t_1,t_2, \cdots,t_{i-1}\}$ and $S_i$ remain unchanged under any work-conserving scheme, one can show that the distribution of $\{t_1,t_2, \cdots, t_i\}$ also remains unchanged under any work-conserving scheme. By induction, we attained that the distribution of the departure process $\omega_D=\{t_1, t_2, \cdots\}$ remains the same under any work-conserving scheme.

Next, we will show that the greedy scheme has the smallest delay  among all work-conserving schemes. According to Eqns. \eqref{eq_1} and \eqref{eq_2},  the expected delay is given by:
\begin{align}
\label{eq_3}
&\mathbb{E}_{\pi}[D]\nonumber\\
=&\lim_{T\rightarrow\infty}\frac{1}{N_T}\sum_{i=1}^{N_T} \mathbb{E}_{\pi} [T_F^i-T_A^i]\nonumber\\
=& \lim_{T\rightarrow\infty} \frac{1}{N_T}\sum_{i=1}^{N_T}\int_{\omega_{A}} \int_{\omega_{D}} \!\!(T_{F:\pi}^i\!-\!T_A^i)f(\omega_A)f(\omega_D|\omega_A)d\omega_Ad\omega_D,\nonumber\\
=& \lim_{T\rightarrow\infty} \int_{\omega_{A}} \int_{\omega_{D}}\! \frac{1}{N_T}\!\sum_{i=1}^{N_T}(T_{F:\pi}^i\!-\!T_A^i)f(\omega_A)f(\omega_D|\omega_A)d\omega_Ad\omega_D,
\end{align}
where the expectation is taken over the distribution of the arrival process
of effective chunks $\omega_A$ and the corresponding distribution of the departure process of the effective chunks $\omega_D$. $T_{F:\pi}^i$ denotes the departure time
of the $i-$th request under scheme $\pi$. Notice that the distributions of $\omega_A$ and $\omega_D$ are always the same for any work-conserving scheme.

For the greedy scheme $\pi_{\textsf{greedy}}$, since it serves the requests one by one and each request requires
exactly $k$ effective chunks, the $i$-th request
departs when the $ik$-th effective chunk departs, i.e.,
\begin{align}\label{eq_31}
T_{F:\pi_{\textsf{greedy}}}^i=t_{ik}.
\end{align}

For an alternative work-conserving
scheme $\pi_{\textsf{alt}}$, suppose that the $i$-th arrival request departs when the $a_i$-th effective chunk departs, i.e.,
\begin{align}\label{eq_32}
T_{F:\pi_{\textsf{alt}}}^i=t_{a_{i}}.
\end{align}
The sequence $\{t_{a_{1}}, t_{a_{2}},\cdots,t_{a_{N_T}}\}$ may not be in an increasing order. After sorting, suppose that this sequence becomes $\{t_{b_{1}}, t_{b_{2}},\cdots,t_{b_{N_T}}\}$ such that $t_{b_{i}}<t_{b_{i+1}}$ and $t_{b_{i}}\in\{t_{a_{1}}, t_{a_{2}},\cdots,t_{a_{N_T}}\}$. Therefore, $i$ requests have departed from the system by time $t_{b_{i}}$. Since each request contains $k$ effective chunks, the systems must have downloaded at least $ik$ effective chunks by $t_{b_{i}}$, which tells us that
\begin{align}\label{eq_33}
t_{b_{i}}\geq t_{ik}.
\end{align}

Using \eqref{eq_31}-\eqref{eq_33}, we can attain
\begin{align}
\label{eq_4}
&\frac{1}{N_T}\sum_{i=1}^{N_T}(T_{F:\pi_{\textsf{greedy}}}^i-T_A^i)-\frac{1}{N_T}\sum_{i=1}^{N_T}(T_{F:\pi_{\textsf{alt}}}^i-T_A^i)\nonumber\\
=&\frac{1}{N_T}\sum_{i=1}^{N_T}(T_{F:\pi_{\textsf{greedy}}}^i-T_{F:\pi_{\textsf{alt}}}^i)\nonumber\\
=&\frac{1}{N_T}\sum_{i=1}^{N_T}(t_{ik}-t_{a_{i}})\nonumber\\
=&\frac{1}{N_T}\sum_{i=1}^{N_T}(t_{ik}-t_{b_{i}})\nonumber\\
\leq&0.
\end{align}
Since the distributions of $\omega_A$ and $\omega_D$ are always the same for any work-conserving scheme, \eqref{eq_4} tells us that the greedy scheme achieves the minimum average delay among all the work-conserving schemes, which completes our proof.

%Note that $a_{i}$ may be greater than $a_{j}$ even if $i<j$. For example,
%the second file may depart earlier than the first file under some work-conserving
%scheme. Nonetheless, we can show that the $i$th smallest item in the sequence of $a_i, i\in[1,2,\cdots, N_t]$ should be no less than $ik$.
%
%First, we notice that in the sequence of $a_i, i\in[1,2,\cdots,
%N_t]$, the smallest item should be no less than $k$, since for the first departure file, it requires to download at least $k$ chunks.
%In addition, the second smallest item in this sequence should be no less than $2k$, since for the second
%departure file, it requires another $k$ chunks except the $k$ chunks for
%the first departure file. Thus, by deduction, we know that in the sequence of $a_i, i\in[1,2,\cdots,
%N_t]$, the $i$th smallest item should be no less than $ik$. It immediately yields that $\sum_{i=1}^{N_T} t_{ik}-\sum_{i=1}^{N_T}t_{a_{i}}\leq0$ for any $N_T$ and $a_i$, which completes our proof.

%
% Under the same sample path of the arrival and departure processes of effective chunks,
%we compare the average delay as follows:
%\begin{align}
%\label{eq_4}
%&\frac{1}{N_T}\sum_{i=1}^{N_T}(T_{F:\pi_{\textsf{greedy}}}^i-T_A^i)-\frac{1}{N_T}\sum_{i=1}^{N_T}(T_{F:\pi_{\textsf{alt}}}^i-T_A^i)\nonumber\\
%=&\frac{1}{N_T}\sum_{i=1}^{N_T}(t_{ik}-T_A^i)-\frac{1}{N_T}\sum_{i=1}^{N_T}(t_{a_{i}}-T_A^i)\nonumber\\
%=&\frac{1}{N_T}\Big[\sum_{i=1}^{N_T} t_{ik}-\sum_{i=1}^{N_T}t_{a_{i}}\Big].
%\end{align}
%
%
%
%We will  show that $\sum_{i=1}^{N_T} t_{ik}-\sum_{i=1}^{N_T}t_{a_{i}}\leq0$ for any $N_T$ and $a_i$.
\end{proof}

\begin{figure}[!t]
\centering
\includegraphics[width=1.0\columnwidth]{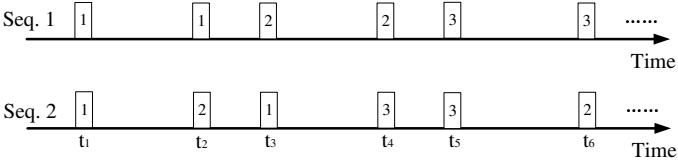}
\caption{Different departure sequences under same sample path}
\label{fig_ser}
\end{figure}

 Fig. \ref{fig_ser} shows an example when $k=2$, i.e., each request requires two effective chunks. In the figure, we plot two same sample paths of effective
chunk departure, where  each rectangle represents an effective
chunk and the number means which request this effective chunk belongs to.
   Sequence $1$ represents the effective chunks belonging under
the greedy scheme, where the chunks depart in order. On the contrary,
 some other work-conserving scheme may have a non-ordered effective chunks departure, such as Sequence $2$. Although for a particular request, its delay
may be greater in greedy scheme, such as request 3 departs later in Sequence
1 than Sequence 2 in the figure, the summation of delay under greedy scheme is always no greater than any other work-conserving scheme, as
$t_2+t_4+t_6\leq t_3+t_5+t_6$ in the example.

 \section{Simulations}\label{simulation}

In this section, we conduct experiments under exponentially distributed downloading time and real traces plotted in Section~\ref{measurement}.

\subsection{Simulation Setup}

 We simulate a system with $L=16$ threads. The downloading requests arrive
 as a Poisson process with parameter $\lambda$. In our simulations, we set
$\lambda=50$, and the number of arrival is set to be 62500.  The expected downloading
 time for any individual thread is assumed to be  $1/\mu$. Thus, the load is given
 by
\begin{align}
\label{eq_8}
\rho=\frac{\lambda}{L\mu}.
\end{align}

Besides the greedy scheme and the sharing scheme, we will also simulate the
round-robin scheme as follows:

\textsf{The Round-robin scheme}: At any time when some threads become idle and the queue is not empty, the
dispatcher allocates these idle threads to all the requests in the queue
in a round-robin way, i.e., the first idle thread to the first request, the
second idle request to the second request, etc..

Notice that the round-robin scheme is also  work-conserving. De facto, the
Round-robin
scheme actually works closely to the greedy for low to medium arrival rates as the probability of finding two requests waiting in the queue is low.
 We take the average
delay over 1000 sample paths for each experiment.

\subsection{Exponential Distribution Case}

First, we show the result under the case of $k=1$. We set $n=L+k-1$. We plot the expected
delay versus different load $\rho$ under the greedy scheme, the round-robin
scheme and the sharing scheme.
As shown in Fig. \ref{fig_expkequalone}, we can see that the delay performance
of the greedy scheme and round-robin scheme are the same, which is much smaller
than the delay of the sharing scheme. The reason is that   both of
the greedy and round-robin scheme are work-conserving, while the sharing
scheme is not. This observation validates our result in Theorem \ref{theorem_1}.

\begin{figure}[!t]
\centering
\includegraphics[width=0.8\columnwidth]{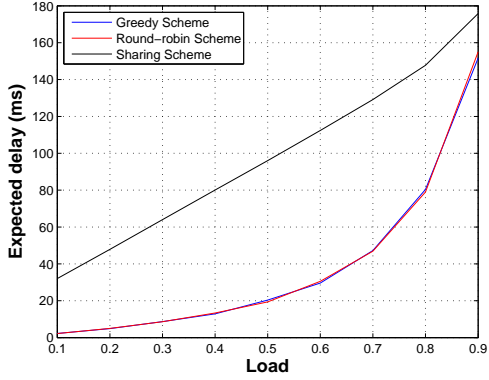}
\caption{Expected delay under greedy, round-robin and sharing schemes when $k=1$}
\label{fig_expkequalone}
\end{figure}

Fig. \ref{fig_expkgreaterone} shows the expected delay under both the greedy
and the round-robin scheme  when  $k=2$ (Since the sharing scheme has a much
worse expected delay, in order to clearly show the gap between the two work-conserving schemes, we do not plot the sharing scheme in the figure). We can see that the greedy
scheme  outperforms the round-robin scheme in terms of expected delay,
which verifies our result in Theorem \ref{theorem_3}. It is indeed  surprising that the greedy scheme outperforms the Round-robin scheme as one can intuitively think that round robin makes more opportunistic use of parallel threads. However, the capacity
regions are the same under both schemes, as we have proven in Theorem \ref{theorem_2}.

In Fig. \ref{fig_expecteddelayvsk}, we fix $\rho=0.1$. It shows that expected delay versus different values of $k$. It can be also noticed that the expected
delay under the greedy scheme is strictly smaller than the expected delay
under the round-robin scheme. And the gap increases as $k$ becomes larger.
The reason is that when $k$ increases, the effective chunk departure
sequence under round-robin scheme has more randomness and becomes more disordered. From the discussion following Theorem \ref{theorem_3},  we know that  the ordered departure of requests under the greedy scheme leads to the smallest delay. Thus, a larger delay
is expected if more disorder occurs.

\begin{figure}[!t]
\centering
\includegraphics[width=0.8\columnwidth]{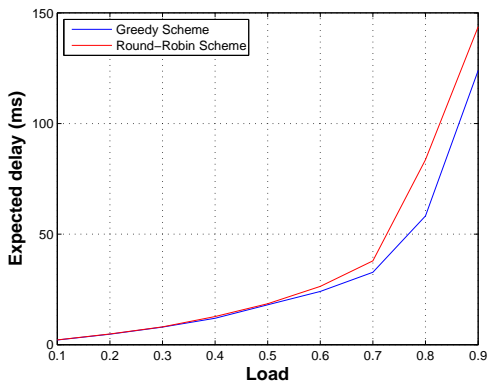}
\caption{Expected delay under greedy and round-robin schemes   when $k=2$}
\label{fig_expkgreaterone}
\end{figure}

\begin{figure}[!t]
\centering
\includegraphics[width=0.8\columnwidth]{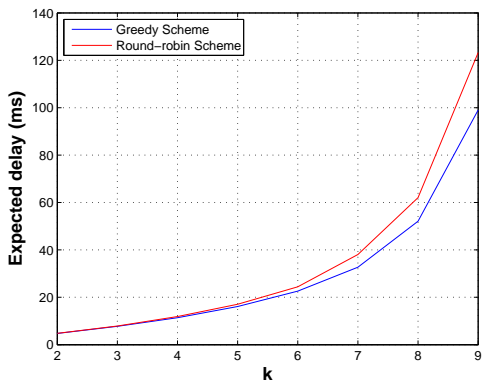}
\caption{Expected delay  versus $k$ under greedy and round-robin schemes}
\label{fig_expecteddelayvsk}
\end{figure}

\subsection{Real Traces}

{Not surprisingly, the simulations using exponential service times for individual threads are in line with the optimality of work-conserving schemes for $k=1$ and the greedy scheme in general for $k\geq 1$ established in the previous sections. As we have indicated in Section~\ref{measurement}, real public clouds such as Amazon S3 have in general non-negligible constant overheads for all files sizes. Thus, our results on delay optimality cannot directly be used for real systems, which calls for further investigations. In \cite{addShengboOriginal}, for instance, it is clearly shown that adding more redundancy reduces the system throughput and suggests that FEC should be adaptive to the system load to improve system performance.

In Figure~\ref{fig_tradeoff}, we plot average latency using real traces collected from Amazon S3 for 1 Mbyte files in North California region in 2012.  In the figure, $k=2$ and each chunk size is 1Mbyte (thus, the original file is of 2Mbyte in size). Thus, to recover the file at least 2 chunks must be downloaded. Using the same traces for chunk service times, we compared non-preemptive first come first serve strategies. Each point in the horizontal axis corresponds to a different strategy that simply dictates how many chunks are requested per file request. As in the previous section, we assume that there are $L = 16$ threads. The file arrival rate is set such that $\rho = 0.05$, i.e., it is low enough to keep the system stable regardless of how many chunks are requested per file. The left-most point in the curve corresponds to the sharing strategy that sends $k=2$ requests per job.  Right-most two points are equivalent to the greedy strategy that allocates all threads to the same job request until $k$ chunks are downloaded for that job. The intermediate strategies trade-offs between maximum sharing vs. all greedy strategy by allocating more and more chunk requests for the same job. As it is clear in the figure that the average system delay first reduces and than increases until the scheduling policies converge to the greedy strategy. In this particular case, downloading 10 chunks per job (in a non-work-conserving fashion!) becomes the best strategy among the fixed FEC rate strategies. Thus, even when the system is operating within the capacity region the queueing delay penalty can wipe out the benefits of service time improvements achieved with using more number of redundant requests. We plot the service time (i.e., the time it takes from the time one of the threads starts serving the first chunk of an object until the k-th chunk of the same object is completely downloaded) in the same figure. Clearly, there is diminishing return for improving the service time with more redundant requests. In contrast, the queueing delay gets worse with higher redundancy.

\begin{figure}[!t]
\begin{center}
\includegraphics[width=0.95\linewidth]{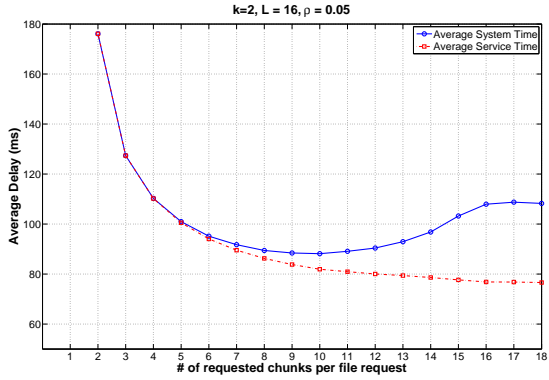}
\caption{Greedy scheme is not delay-optimal in general over real clouds. }
\label{fig_tradeoff}
\end{center}
\end{figure}

\begin{figure}[!t]
\begin{center}
\includegraphics[width=0.95\linewidth]{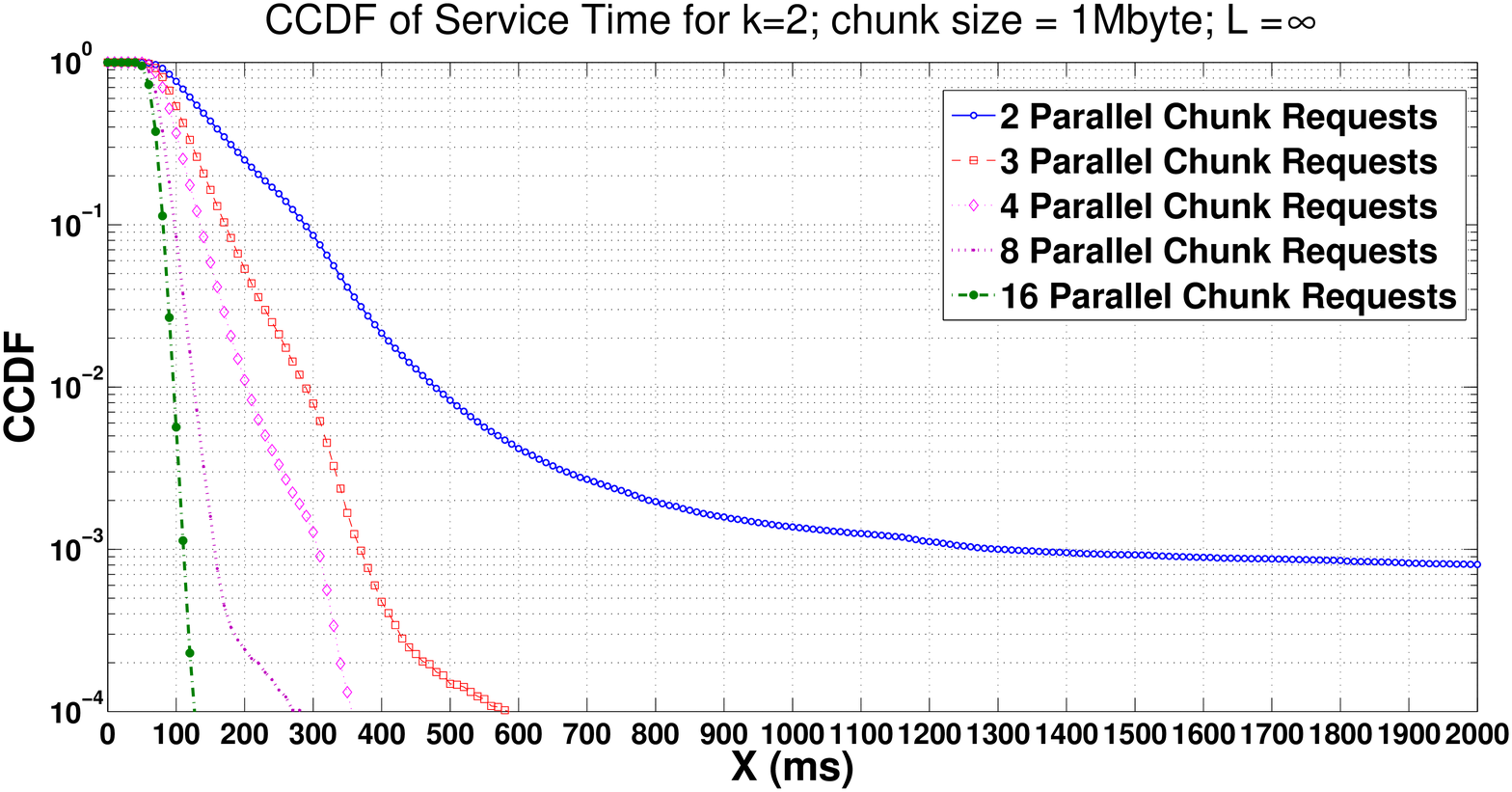}
\caption{Complementary cumulative distribution function  (CCDF) of service time with unlimited bandwidth for a 2Mbyte file and k=2. }
\label{unlimited_bw}
\end{center}
\end{figure}

To illustrate how sending redundant requests improve the service time, we plot the CCDF for a 2Mbyte file. As before we set the chunk size to 1Mbyte (i.e., $k=2$), but this time assumed that $L=\infty$ to provide a lower bound on service time distribution with unlimited bandwidth. Fig.~\ref{unlimited_bw} shows CCDF curves when 2, 3, 4, 8 and 16 encoded chunks are requested simultaneously. The individual chunk service times are simulated using the same Amazon S3 traces as in Fig.~\ref{fig_tradeoff}. Notice that the distribution gets sharpened and when 16 chunks are requested in parallel, the shape comes close to a deterministic service time. Under deterministic service times, we know that the sharing scheme is the delay optimal strategy. Therefore, even under relatively high bandwidth conditions (e.g., $L >> n$), adding redundancy beyond a certain threshold should not bring any benefit and a sharing strategy that allocates a limited number of threads to each file request should perform better than a pure greedy strategy that allocates all threads to the same job request. These insights suggest that the theoretical insights should be applied with care only in situations when the exponential assumption is reasonable.

}

%We here adopt the real traces of  downloading delay that are plotted in Section \ref{measurement}.

%In Fig. \ref{fig_cuttingspread}, we plot the read delay spread  under different
%level of parallel downloading when $k=1$. The red curve in the figure represents the original CCDF, that is,
%without making use of FEC redundancy.  The other four curves represent the
%CCDF of read delay using different number of parallel threads to retrieve
%the file, which
%are 2, 4, 8, and 16, respectively. We can see that the delay spread can be significantly reduced by leveraging the FEC redundancy. And as the level of parallelism   increases, the delay spread becomes narrower,  resulting
%in a smaller latency.

%\begin{figure}[!t]
%\centering
%\includegraphics[width=0.8\columnwidth]{fig/cuttingspread.eps}
%\caption{CCDF of read delay under different level of parallel downloading}
%\label{fig_cuttingspread}
%\end{figure}

%
%
 \section{Conclusion}\label{conclusion}

In this paper, we study delay performance of  downloading
data
from cloud storages by leveraging multiple parallel    threads, assuming that the
data in the clouds has been encoded using FEC codes. This leads to a new
queueing model. {We roughly approximate the downloading time for any individual thread as
an \textit{i.i.d. }exponentially distributed random variable.}
We show that the any work-conserving scheme is delay-optimal {among all on-line scheduling schemes} when $k=1$. When $k>1$, we prove that a simple greedy scheme
is delay optimal {among all on-line scheduling schemes}.
We validate our results through simulations with exponentially distributed
service time. {There are two interesting important open directions: One is to analyze the data retrieving delay of different scheduling schemes in the non-exponential case. Some initial efforts were made recently in \cite{addShengboOriginal,Joshi-2012,Guanfeng2014,shah-Allerton-2013}.
Another open problem is to analyze the data retrieving delay of different scheduling schemes when the condition $n \geq L+k-1$ does not hold.}

%By using real traces, we show that making use of inherent redundancy introduced by FEC can indeed reduce the delay and significantly cut the delay spread.

\bibliographystyle{IEEEtran}
\small
\bibliography{ref}

\begin{thebibliography}{10}
\providecommand{\url}[1]{#1}
\csname url@samestyle\endcsname
\providecommand{\newblock}{\relax}
\providecommand{\bibinfo}[2]{#2}
\providecommand{\BIBentrySTDinterwordspacing}{\spaceskip=0pt\relax}
\providecommand{\BIBentryALTinterwordstretchfactor}{4}
\providecommand{\BIBentryALTinterwordspacing}{\spaceskip=\fontdimen2\font plus
\BIBentryALTinterwordstretchfactor\fontdimen3\font minus
  \fontdimen4\font\relax}
\providecommand{\BIBforeignlanguage}[2]{{%
\expandafter\ifx\csname l@#1\endcsname\relax
\typeout{** WARNING: IEEEtran.bst: No hyphenation pattern has been}%
\typeout{** loaded for the language `#1'. Using the pattern for}%
\typeout{** the default language instead.}%
\else
\language=\csname l@#1\endcsname
\fi
#2}}
\providecommand{\BIBdecl}{\relax}
\BIBdecl

\bibitem{CloudStorageReport}
``Public/private cloud storage market,''
  \url{http://www.marketsandmarkets.com/PressReleases/cloud-storage.asp}.

\bibitem{wuala}
``Wuala,'' http://www.wuala.com/.

\bibitem{addShengboOriginal}
G.~Liang and U.~Kozat, ``{FAST CLOUD: Pushing the Envelope on Delay Performance
  of Cloud Storage with Coding},'' \emph{IEEE/ACM Trans. Networking}, Nov 2013,
  preprint.

\bibitem{Garfinkel07anevaluation}
S.~L. Garfinkel, ``An evaluation of {AmazonÕs} grid computing services: {EC2},
  {S3} and {SQS},'' Harvard University, Tech. Rep., 2007.

\bibitem{lu-jiq-2011}
Y.~Lu, Q.~Xie, G.~Kliot, A.~Geller, J.~Larus, and A.~Greenberg,
  ``Join-idle-queue: A novel load balancing algorithm for dynamically scalable
  web services,'' \emph{Performance Evaluation}, 2011.

\bibitem{dimakis-2010}
A.~G. Dimakis, P.~B. Godfrey, Y.~Wu, M.~Wainwright, and K.~Ramchandran,
  ``Network coding for distributed storage systems,'' \emph{IEEE Trans. Inf.
  Theory, vol. 56, no. 9}, Sept. 2010.

\bibitem{dimakis-2009}
Y.~Wu and A.~G. Dimakis, ``Reducing repair traffic for erasure coding-based
  storage via interference alignment,'' \emph{IEEE Intl Symp. on Information
  Theory (ISIT)}, 2009.

\bibitem{rashmi-2010}
K.~V. Rashmi, N.~B. Shah, P.~V. Kumar, and K.~Ramchandran, ``Explicit and
  optimal exact-regenerating codes for the minimum-bandwidth point in
  distributed storage,'' \emph{IEEE Intl Symp. on Information Theory (ISIT)},
  2010.

\bibitem{leong-2009}
D.~Leong, A.~G. Dimakis, and T.~Ho, ``Distributed storage allocation
  problems,'' \emph{Proceedings of the Workshop on Network Coding, Theory, and
  Applications (NetCod)}, 2009.

\bibitem{coding-survey-dimakis}
A.~G. Dimakis, K.~Ramchandran, Y.~Wu, and C.~Suh, ``A survey on network codes
  for distributed storage,'' \emph{Proceedings of the IEEE, vol. 99, no. 3},
  March 2011.

\bibitem{alex-delay-linenet09}
T.~Dikaliotis, A.~G. Dimakis, T.~Ho, and M.~Effros, ``On the delay of network
  coding over line networks,'' \emph{IEEE International Symposium on
  Information Theory (ISIT)}, 2009.

\bibitem{ferner-2012}
U.~J. Ferner, M.~Medard, and E.~Soljanin, ``Toward sustainable networking:
  Storage area networks with network coding,'' http://arxiv.org/abs/1205.3797,
  2012.

\bibitem{ali-poe-2012}
A.~ParandehGheibi, M.~Medard, A.~Ozdaglar, and S.~Shakkottai, ``Avoiding
  interruptions: A {QoE} reliability function for streaming media
  applications,'' \emph{IEEE J. Sel. Areas Commun.}, vol.~29, no.~5, pp.
  1064--1074, May 2011.

\bibitem{RRighter}
J.~H. Kim, H.-S. Ahn, and R.~Righter, ``{Managing queues with heterogeneous
  servers},'' \emph{Journal of Applied Probability}, vol.~48, pp. 435--452,
  2011.

\bibitem{RRighter1}
Y.~Kim, R.~Righter, and R.~Wolff, ``{Grid scheduling with NBU service times},''
  \emph{Operations Research Letters}, vol.~38, pp. 502--504, 2010.

\bibitem{huang-isit-2012}
L.~Huang, S.~Pawar, H.~Zhang, and K.~Ramchandran, ``Codes can reduce queueing
  delay in data centers,'' \emph{Proceedings of IEEE International Symposium on
  Information Theory (ISIT)}, July 2012.

\bibitem{shah-mdsq-2012}
N.~B. Shah, K.~Lee, and K.~Ramchandran, ``The {MDS} queue: Analysing latency
  performance of codes,'' http://arxiv.org/abs/1211.5405, to appear in
  \emph{IEEE ISIT}, 2014.

\bibitem{Joshi-2012}
G.~Joshi, Y.~Liu, and E.~Soljanin, ``On the delay-storage trade-off in content
  download from coded distributed storage systems,''
  http://arxiv.org/abs/1305.3945, 2013.

\bibitem{Guanfeng2014}
G.~Liang and U.~C. Kozat, ``On throughput-delay optimal access to storage
  clouds via load adaptive coding and chunking,'' in \emph{IEEE INFOCOM}, 2014.

\bibitem{shah-Allerton-2013}
N.~B. Shah, K.~Lee, and K.~Ramchandran, ``When do redundant requests reduce
  latency?'' in \emph{Allerton Conference}, 2013.

\end{thebibliography}

\end{document}